\documentclass[11pt]{article}
\usepackage{etex}

\usepackage{amsthm}
\usepackage{amsfonts}
\usepackage{amsmath}
\usepackage{amssymb}
\usepackage{color}
\usepackage{easybmat}
\usepackage{blkarray}
\usepackage{hyperref}
\usepackage{comment}
\usepackage{accents}

\usepackage[T2A]{fontenc}

\usepackage{mathrsfs}
\renewcommand{\mathcal}[1]{\mathscr{#1}}

\usepackage{marvosym}

\usepackage{tikz}
\usetikzlibrary{patterns,calc,decorations.pathmorphing,positioning,arrows,fit}

\usepackage{theory}
\usepackage{theorems}

\renewcommand{\phi}{\varphi}

\newcommand{\bb}[1]{\mathbb{#1}}

\newcommand{\nm}[1]{\textit{#1}}

\date{}
\title{Circuit Complexity Meets \\ Ontology-Based Data Access}

\author{Vladimir V. Podolskii\footnote{This work is supported by the Russian Science Foundation under grant
14-50-00005 and performed in Steklov Mathematical Institute of Russian
Academy of Sciences.
The final publication to appear in proceedings of CSR 2015, LNCS 9139, Springer.}\\[3pt]
\small Steklov Mathematical Institute, Moscow, Russia\\
\small National Research University Higher School of Economics, Moscow, Russia\\
\small   \href{mailto:podolskii@mi.ras.ru}{podolskii@mi.ras.ru}
}

\date{}
\begin{document}

\maketitle

\begin{abstract}
Ontology-based data access is an approach to organizing access to a database augmented with a logical theory.
In this approach query answering proceeds through a reformulation of a given query into a new one
which can be answered without any use of theory.
Thus the problem reduces to the standard database setting.

However, the size of the query may increase substantially during the reformulation.
In this survey we review a recently developed framework on proving lower and upper bounds on the size of this reformulation
by employing methods and results from Boolean circuit complexity.
\end{abstract}

\section{Introduction}

Ontology-based data access is an approach to storing and accessing data in a database\footnote{We use the word ``database'' in a wide informal sense, that is a database is an organized collection of data.}. In this approach the database is augmented with a first-order logical theory, that is
the database is viewed as a set of predicates on elements (entities) of the database
and the theory contains some universal statements about these predicates.

The idea of augmenting data with a logical theory has been around since at least 1970s (the Prolog programming language, for example, is in this flavor~\cite{Kowalski:1988:EYL:35043.35046}).
However, this idea had to constantly overcome implementational issues.
The main difficulty is that if the theory accompanying the data is too strong,
then even standard algorithmic tasks become computationally intractable.

One of these basic algorithmic problems will be of key interest to us, namely the query answering problem.
A query to a database seeks for all elements in the data with certain properties.
%The decision version of this problem is to find out whether there is such an element.
In case the data is augmented with a theory, query answering cannot be handled directly with the same methods as for usual databases
and new techniques are required.

Thus, on one hand, we would like a logical theory to help us in some way and, on the other hand, we
need to avoid arising computational complications.

Ontology-based data access (OBDA for short) is a recent approach in this direction developed since around 2005~\cite{CDLLR05,DFKM*08,HMAM*08,PLCD*08}.
Its main purpose is to help maintaining large and distributed
data and make the work with the data more user-friendly.
The logical theory helps in achieving this goal by allowing one to create a convenient language for queries,
hiding details of the structure of the data source, supporting queries to distributed and heterogeneous data sources.
Another important property is that data does not have to be complete.
Some of information may follow from the theory and not be presented in the data explicitly.

A key advantage of OBDA is that to achieve these goals, it is often enough in practice to supplement the data with a rather primitive theory.
This is important for the query answering problem:
the idea of OBDA from the algorithmic point of view is not to develop a new machinery, but to reduce
query answering with a theory to the standard database query answering and use the already existing machinery.

The most standard approach to this is to first reformulate a given query in such a way that the answer to the new query does not depend on the theory anymore.
This reformulation is usually called a \emph{rewriting} of the query.
The rewriting should be the same for any data in the database.
Once the rewriting is built we can apply standard methods of database theory.
Naturally, however, the length of the query typically increases during the reformulation
and this might make this approach (at least theoretically) inefficient.

The main issue we address in this survey is how large the rewriting can be compared to the size of the original query.
Ideally, it would be nice if the size of the rewriting is polynomial in the size of the original query.
In this survey we will discuss why rewritings can grow exponentially in some cases and how Boolean circuit complexity helps us to obtain
results of this kind.

In this survey we will confine ourselves to data consisting only of unary and binary predicates over the database elements.
If data contains predicates of larger arity, the latter can be represented via binary predicates.
Such representations are called \emph{mappings} in this field and there are several ways for doing this.
We leave the discussion of mappings aside and refer the reader to~\cite{chapter_KZ} and references therein.
We call a data source with unary and binary predicates augmented with a logical theory a \emph{knowledge base}.

As mentioned above, in OBDA only very restricted logical theories are considered.
There are several standard families of theories, including OWL~2~QL~\cite{profiles,CDLLR07,ACKZ09} and several fragments of Datalog$^{\pm}$~\cite{DBLP:journals/ws/CaliGL12,DBLP:conf/ijcai/BagetLMS09,CaliGP10,DBLP:journals/ai/CaliGP12}.
The lower bounds on the size of rewritings we are going to discuss work for even weaker theories contained in all families mentioned above.
The framework we describe also allows one to prove upper bounds on the size of the rewritings that work for theories given in OWL~2~QL.
We will describe the main ideas for obtaining upper bounds, but will not discuss them in detail.

To give a complete picture of our setting, we need also to discuss the types of queries and rewritings we consider.
The  standard type of queries (as a logical formulas) considered in this field is conjunctive queries,
i.e. conjunctions of atomic formulas prefixed by existential quantifiers.
%The reason for this is that, on one hand, these queries are basically all that is used in practice, and on the other hand,
%they have the simplest possible structure.
In this survey we will discuss only this type of queries.

As for rewritings, it does not make sense to consider conjunctive formulas as rewritings, since their expressive power is rather poor.
The simplest type of rewritings that is powerful enough
to provide a rewriting for every query is a DNF-rewriting, which is a disjunction of conjunctions with an existential quantifiers prefix.
However, it is not hard to show (see~\cite{GKKPSZ14AI}) that this type of rewriting may be exponentially larger than the original query.
More general standard types of rewritings are first-order (FO-) rewritings, where a rewriting can be an arbitrary first-order formula,
positive existential (PE-) rewritings, which are first-order formulas containing only existential quantifiers and no negations
(this type of rewritings is motivated by its convenience for standard databases), and the nonrecursive datalog rewriting, which are not first-order formulas but rather are constructed in a more circuit-flavored way (see Section~\ref{sec:rewritings_intro} for details).

For these more general types of rewritings it is not easy to see how the size of the rewriting grows in size of the original query.
The progress on this question started with the paper~\cite{KKZ11DL}, where it was shown
that the polynomial size FO-rewriting cannot be constructed in polynomial time, unless $\P=\NP$. Soon after that, the approach of that paper was extended in~\cite{KKPZ12ICALP,GKKPSZ14AI} to give a much stronger result: not only there is no way to construct a FO-rewriting in polynomial time, but even there is no polynomial size FO-rewriting, unless $\NP \subseteq \P/\poly$. It was also shown (unconditionally!) in~\cite{KKPZ12ICALP,GKKPSZ14AI} that there are queries and theories for which the shortest PE- and NDL-rewritings are exponential in the size of the original query. They also obtained an exponential separation between PE- and NDL-rewritings and a superpolynomial separation between PE- and FO-rewritings.

These results were obtained in~\cite{KKPZ12ICALP,GKKPSZ14AI}
by reducing the problems of lower bounding the rewriting size to some problems in computational complexity theory.
Basically, the idea is that we can encode a Boolean function $f \in \NP$ into a query $q$ and design the query and the theory in such a way that a FO-rewriting of $q$ will provide us with a Boolean formula for $f$, a PE-rewriting of $q$ will correspond to a monotone Boolean formula, and an NDL-rewriting --- to a monotone Boolean circuit. Then by choosing an appropriate $f$ and applying known results from circuit complexity theory, we can deduce the lower bounds on the sizes of the rewritings.

The next step in this line of research was to study the size of rewritings for restricted types of queries and knowledge bases. A natural subclass of conjunctive queries is the class of tree-like queries. To define this class, for a given query consider a graph whose vertices are the variables of the query and an edge connects two variables if they appear in the same predicate in the query. We say that a query is a \emph{tree-like} if this graph is a tree.
A natural way to restrict theories of knowledge bases is to consider their depth. Informally, the theory is of depth $d$ if, starting with a data and generating all new objects whose existence follows from the given theory, we will not obtain in the resulting underlying graph any sequences of new objects of length greater than $d$. These kinds of restrictions on queries and theories are motivated by practical reasons: they are met in the vast majority of applications of knowledge bases. On the other hand, in papers~\cite{KKPZ12ICALP,GKKPSZ14AI} non-constant depth theories were used to prove lower bounds on the size of rewritings.

Subsequent papers~\cite{LICS14,DL14} managed to describe a complete picture of the sizes of the rewritings in restricted cases described above. To obtain these results, they determined, for each case mentioned above, the class of Boolean functions $f$ that can be encoded by queries and theories of the corresponding types. This establishes a close connection between ontology-based data access and various classes in Boolean circuit complexity. Together with known results in Boolean circuit complexity, this connection allows one to show various lower and upper bounds on the sizes of rewritings in all described cases.
The precise formulation of these results is given in Section~\ref{sec:proof_idea}.

To obtain their results, \cite{LICS14,DL14} also introduced a new intermediate computational model, the \emph{hypergraph programs}, which might be of independent interest. A hypergraph program consists of a hypergraph whose vertices are labeled by Boolean constants, input variables $x_1, \ldots, x_n$ or their negations. On a given input $\vec{x} \in \zo^n$, a hypergraph program outputs $1$ iff all its vertices whose labels are evaluated to $0$ on this input can be covered by a set of disjoint hyperedges. We say that a hypergraph program computes $f \colon \zo^ n \to \zo$ if it outputs $f(\vec{x})$ on every input $\vec{x} \in \{0,1\}^n$. The size of a hypergraph program is the number of vertices plus the number of hyperedges in it.

Papers~\cite{LICS14,DL14} studied the power of hypergraph programs and their restricted versions. As it turns out, the class of functions computable by general hypergraph programs of polynomial size coincides with $\NP/\poly$~\cite{LICS14}. The same is true for hypergraph programs of degree at most $3$, that is for programs in which the degree of each vertex is bounded by $3$. The class of functions computable by polynomial size hypergraph programs of degree at most $2$ coincides with $\NL/\poly$~\cite{LICS14}. Another interesting case is the case of \emph{tree hypergraph programs} which have an underlying tree and all hyperedges consist of subtrees. Tree hypergraph programs turn out to be equivalent to $\SAC^{1}$ circuits~\cite{DL14}. If the underlying tree is a path, then polynomial size hypergraph programs compute precisely the functions in $\NL/\poly$~\cite{DL14}.

The rest of the survey is organized as follows.
In Section~\ref{sec:circuits_intro} we give the necessary definitions from Boolean circuit complexity.
In Section~\ref{sec:rewritings_intro} we give the necessary definitions and basic facts on knowledge bases.
In Section~\ref{sec:proof_idea} we describe the main idea behind the proofs of bounds on the size of the rewritings.
In Section~\ref{sec:hypergraph_intro} we introduce hypergraph programs and explain how they help to bound the size of the rewritings.
In Section~\ref{sec:hypergraph_complexity} we discuss the complexity of hypergraph programs.

\section{Boolean circuits and other computational models} \label{sec:circuits_intro}

In this section we provide necessary information on Boolean circuits, other computational models and related complexity classes.
For more details see~\cite{Jukna12}.

\emph{A Boolean circuit} $C$ is an acyclic directed graph. Each vertex of the graph is labeled by either a variable among $x_1, \ldots, x_n$, or a constant $0$ or 1, or a Boolean function $\neg$, $\wedge$ or $\vee$. Vertices labeled by variables and constants have in-degree $0$, vertices labeled by $\neg$ have in-degree $1$,
vertices labeled by $\wedge$ and $\vee$ have in-degree $2$. Vertices of a circuit are called \emph{gates}.
Vertices labeled by variables or constants are called \emph{input gates}.
For each non-input gate $g$ its \emph{inputs} are the gates which have out-going edges to $g$.
One of the gates in a circuit is labeled as an \emph{output gate}.
Given $\vec{x} \in \{0,1\}^n$, we can assign the value to each gate of the circuit inductively.
The values of each input gate is equal to the value of the corresponding variable or constant.
The value of a $\neg$-gate is opposite to the value of its input. The value of a $\wedge$-gate is equal to $1$ iff both its inputs are $1$.
The value of a $\vee$-gate is $1$ iff at least one of its inputs is $1$. The value of the circuit $C(\vec{x})$ is defined as the value of its output gate on $\vec{x} \in \{0,1\}^n$.
A circuit $C$ \emph{computes} a function $f \colon \{0,1\}^n \to \{0,1\}$ iff $C(\vec{x})=f(\vec{x})$ for all $\vec{x} \in \{0,1\}^n$. The size of a circuit is the number of gates in it.

The number of inputs $n$ is a parameter. Instead of individual functions, we consider sequences of functions $f = \{f_n\}_{n \in \bb{N}}$, where $f_n \colon \zo^n \to \zo$. A sequence of circuits $C = \{C_n\}_{n \in \bb{N}}$ \emph{computes} $f$ iff $C_n$ computes $f_n$ for all $n$.
From now on, by a Boolean function or a circuit we always mean a sequence of functions or circuits.

A \emph{formula} is a Boolean circuit such that each of its gates has fan-out $1$.
A Boolean circuit is \emph{monotone} iff there are no negations in it.
It is easy to see that any monotone circuit computes a monotone Boolean function and,
on the other hand, any monotone Boolean function can be computed by a monotone Boolean circuit.

A circuit $C$ is a \emph{polynomial size circuit} (or just polynomial circuit) if there is a polynomial $p \in \bb{Z}[x]$ such that the size of $C_n$ is at most $p(n)$.

Now we are ready to define several complexity classes based on circuits.
A Boolean function $f$ lies in the class $\P/\poly$ iff there is a polynomial size circuit $C$ computing $f$.
A Boolean function $f$ lies in the class $\NC^1$ iff there is a polynomial size formula $C$ computing $f$.
A Boolean function $f$ lies in the class $\NP/\poly$ iff there is a polynomial $p(n)$ and a polynomial size circuit $C$ such that for all $n$ and for all $\vec{x} \in \{0,1\}^n$
\begin{equation}\label{eq:NP_poly}
f(\vec{x}) = 1 \ \Leftrightarrow \ \exists \vec{y} \in \zo^{p(n)} \ C_{n+p(n)}(\vec{x},\vec{y}) = 1.
\end{equation}
Complexity classes $\P/\poly$ and $\NP/\poly$ are nonuniform analogs of $\P$ and $\NP$.

We can introduce monotone analogs of $\P/\poly$ and $\NC^1$ by considering only monotone circuits or formulas.
In the monotone version of $\NP/\poly$ it is only allowed to apply negations directly to $\vec{y}$-inputs.

The \emph{depth} of a circuit is the length of the longest directed path from an input to the output of the circuit.
It is known that $f \in \NC^1$ iff $f$ can be computed by logarithmic depth circuit~\cite{Jukna12}.
%A Boolean function $f$ is in the complexity class $\NC^1$ if there is a polynomial size logarithmic depth circuit computing $f$.
By $\SAC^1$ we denote the class of all Boolean functions $f$ computable by a polynomial size logarithmic depth circuit
such that $\vee$-gates are allowed to have arbitrary fan-in and all negations are applied only to inputs of the circuit~\cite{circuit}.

A \emph{nondeterministic branching program} $P$ is a directed graph $G=(V,E)$, with edges labeled by Boolean constants, variables $x_1, \ldots, x_n$ or their negations.
There are two distinguished vertices of the graph named $s$ and $t$.
On an input $\vec{x} \in \zo^n$ a branching program $P$ outputs $P(\vec{x})=1$ iff there is a path from $s$ to $t$ going through edges whose labels evaluate to $1$. A nondeterministic branching program $P$ \emph{computes} a function $f \colon \zo^n \to \zo$ iff for all $\vec{x} \in \zo^n$ we have $P(\vec{x})=f(\vec{x})$. The \emph{size} of a branching program is the number of its vertices plus the number of its edges $|V| + |E|$.
A branching program is \emph{monotone} if there are no negated variables among labels.

Just as for the functions and circuits, from now on by a branching program we mean a sequence of branching programs $P_n$ with $n$ variables for all $n \in \bb{N}$.

A branching program $P$ is a \emph{polynomial size branching program} if there is a polynomial $p \in \bb{Z}[x]$ such that the size of $P_n$ is at most $p(n)$.

A Boolean function $f$ lies in the class $\NBP$ iff there is a polynomial size branching program $P$ computing $f$.
It is known that $\NBP$ coincides with nonuniform analog of nondeterministic logarithmic space $\NL$, that is $\NBP= \NL/\poly$~\cite{Jukna12,Razborov91}.

For every complexity class $\K$ introduced above, we denote by $\mK$ its monotone counterpart.

The following inclusions hold between the classes introduced above~\cite{Jukna12}
\begin{equation} \label{eq:inclusions}
\NC^1 \subseteq \NBP \subseteq \SAC^1 \subseteq \P/\poly \subseteq \NP/\poly.
\end{equation}
It is a major open problem in computational complexity whether any of these inclusions is strict.

Similar inclusions hold for monotone case:
\begin{equation} \label{eq:inclusions_monotone}
\mNC^1 \subseteq \mNBP \subseteq \mSAC^1 \subseteq \mP/\poly \subseteq \mNP/\poly.
\end{equation}

It is also known that $\mP/\poly \neq \mNP/\poly$~\cite{Razborov85,AlonB87} and $\mNBP \neq \mNC^1$~\cite{Karchmer88}.
We will use these facts to prove lower bounds on the rewriting size.

\section{Theories, queries and rewritings} \label{sec:rewritings_intro}

In this survey a data source is viewed as a first-order theory. It is not an arbitrary theory and must satisfy some restrictions, which we specify below.

First of all, in order to specify the structure of data, we need to fix a set of predicate symbols in the signature.
Informally, they correspond to the types of information the data contains.
We assume that there are only unary and binary predicates in the signature.
The data itself consists of a set of objects (entities) and of information on them.
%Recall that the database does not have to be complete, that is it might lack some of the information.
Objects in the data correspond to constants of the signature.
%, that is for each of them we introduce a constant in our signature.
The information in the data corresponds to closed atomic formulas,
that is predicates applied to constants.
These formulas constitute the theory corresponding to the data.
We denote the resulting set of formulas by $D$ and the set of constants in the signature by $\Delta_D$.

%Up to this moment we have introduced a signature and a part of the theory.
We denote the signature (the set of predicate symbols and constants) by $\Sigma$.
Thus, we translated a data source into logical terms.
%The signature $\Sigma$ should consist of a constant for any element of $M$ and of predicate symbols. Moreover, all predicates should be unary or binary.
%Elements of the database correspond to elements of the model. The predicates correspond to the information contained in the database.
%So, the database contains statements of the form $P(a,b)$ or $Q(a)$, where $a,b$ are elements of the model $M$ and $P,Q \in \Sigma$ are predicates.
%Thus, we include all the statements, contained in the database to the theory, which we will denote by $A$.
%Now we have reformulated standard databases in logical terms, there is nothing `'ontology'' about them.
%On the other hand, currently the theory is very poor, it consists of atom formulas.
To obtain knowledge base, we introduce more complicated formulas into the theory.
The set of these formulas will be denoted by $T$ and called an \emph{ontology}.
We will describe which formulas can be presented in $T$ a bit later.
%Just like predicate symbols in our signature $T$ determines the structure of the database and thus will be fixed.
The theory $D \cup T$ is called a \emph{knowledge base}.
Predicate symbols and the theory $T$ determines the structure of the knowledge base and thus will be fixed.
Constants $\Delta_D$ and atomic formulas $D$, on the other hand, determine the current containment of the data,
so they will be varying.

As we mentioned in Introduction, we will consider only conjunctive queries. That is, a query is a formula of the form
$$
q(\vec{x}) = \exists \vec{y} \phi(\vec{x},\vec{y}),
$$
where $\phi$ is a conjunction of atomic formulas (or atoms for short).
For simplicity we will assume that $q$ does not contain constants from $\Delta_D$.

What does the query answering mean for standard data sources without ontology?
It means that there are values for $\vec{x}$ and $\vec{y}$ among $\Delta_D$ such that the query becomes true on the given data $D$.
That is, we can consider a model $I_D$ corresponding to the data $D$. The elements of the model $I_D$ are constants in $\Delta_D$
and the values of predicates in $I_D$ is given by formulas in $D$.
That is, a predicate $P \in \Sigma$ is true on $\vec{a}$ from $\Delta_D$ iff $P(\vec{a}) \in D$.
The tuple of elements $\vec{a}$ of $I_D$ is an answer to the query $q(\vec{x})$ if
$$
I_D \models \exists \vec{y} \phi(\vec{a},\vec{y}).
$$

Let us go back to our setting. Now we consider data augmented with a logical theory. This means that we do not have a specific model.
Instead, we have a theory and we need to find out whether the query is satisfied in the theory.
That is, the problem we are interested in is, given a knowledge base $D \cup T$ and a query $q(\vec{x})$, to find $\vec{a}$ in $\Delta_D$ such that
$$
D \cup T \models q(\vec{a}).
$$
If $\vec{x}$ is an empty tuple of variables, then the answer to the query is `yes' or `no'. In this case we say that the query is \emph{Boolean}.

The main approach to solving the query answering problem is to first reformulate the query in such a way that the answer to the new query does not depend on the theory $T$ and then apply the machinery for standard databases. This leads us to the following definition. A first-order formula $q^\prime(\vec{x})$ is called a \emph{rewriting} of $q(x)$ w.r.t. a theory $T$ if
\begin{equation} \label{eq:rewriting}
D \cup T \models q(\vec{a})\  \Leftrightarrow \ I_D \models q^{\prime}(\vec{a})
\end{equation}
for all $D$ and for all $\vec{a}$. We emphasize that on the left-hand side in~\eqref{eq:rewriting} the symbol `$\models$' means logical consequence from a theory,
while on the right-hand side it means truth in a model.

We also note that in~\eqref{eq:rewriting} only predicate symbols in $\Sigma$ and the theory $T$ are fixed. The theory $D$ (and thus, the set of constants in the signature) may vary, so the rewriting should work for any data $D$. Intuitively, this means that the structure of the data is fixed in advance and known, and the current content of a knowledge base may change. We would like the rewriting (and thus the query answering approach) to work no matter how the data change.

%Overall, database correspond to the signature and the theory. Predicates of the signature and $T$-part of the theory defines the structure of the database and are fixed. Constants of the signature and $A$-part represent the information contained in the database and thus can vary. Thus, we are interested in the rewritings for fixed $T$, but for varying $A$.

What corresponds to a model of the theory $D \cup T$? Since the data $D$ is not assumed to be complete,
it is not a model.
A model correspond to the content of the ``real life'' complete data, which extends the data $D$.
We assume that all formulas of the theory hold in the model, that is all information in the knowledge base (including formulas in $T$) is correct.

However, if we allow to use too strong formulas in our ontology, then the problem of query answering will become algorithmically intractable. So we have to allow only very restricted formulas in $T$. On the other hand, for the practical goals of OBDA also only very simple formulas are required.

There are several ways to restrict theories in knowledge bases.
We will use the one that fits all most popular restrictions.
Thus our lower bounds will hold for most of the considered settings.
As for the upper bounds, we will not discuss them in details, however, we mention that they hold for substantially stronger theories and cover OWL 2 QL framework~\cite{profiles}.

Formulas in the ontology $T$ are restricted to the following form
\begin{equation} \label{eq:tgd}
\forall x (\phi(x) \rightarrow \exists y \psi(x,y)),
\end{equation}
where $x$ and  $y$  are (single) variables, $\phi$ is a unary predicate and $\psi(x,y)$ is a conjunction of atomic formulas.

It turns out that if $T$ consists only of formulas of the form~\eqref{eq:tgd}, then the rewriting is always possible.
The (informal) reason for this is that in this case there is always a universal model $M_{D}$ for given $D$ and $T$.
\begin{theorem} \label{thm:universal}
For all theories $D,T$ such that $T$ consists of formulas of the form~\eqref{eq:tgd} there is a model $M_D$ such that
$$
D\cup T \models q(\vec{a}) \ \Leftrightarrow\ M_D \models q(\vec{a})
$$
for any conjunctive query $q$ and any $\vec{a}$.
\end{theorem}
\begin{remark}
Note that the model $M_D$ actually depends on both $D$ and $T$. We do not add $T$ as a subscript since
in our setting $T$ is fixed and $D$ varies.
\end{remark}

The informal meaning of this theorem is that for ontologies $T$ specified by~\eqref{eq:tgd} there is always the most general model.
More formally, for any other model $M$ of $D \cup T$ there is a homomorphism from the universal model $M_D$ to $M$.
We provide a sketch of the proof of this theorem. For us it will be useful to see how the model $M_D$ is constructed.

\begin{proof}[Proof sketch]
The informal idea for the existence of the universal model is that we can reconstruct it from the constants presented in the data $D$.
Namely, first we add to $M_D$ all constants in $\Delta_D$ and we let all atomic formulas in $D$ to be true on them.
Next, from the theory $T$ it might follow that some other predicates should
hold on the constants in $\Delta_D$. We also let them to be true in $M_D$.
What is more important, formulas in $T$ might also imply the existence of new elements related to constants (the formula~\eqref{eq:tgd} implies, for elements $x$ that satisfy $\varphi(x)$, the existence of a new element $y$).
We add these new elements to the model and extend predicates on them by deducing everything that follows from $T$.
Next, $T$ may imply the existence of further elements that are connected to the ones obtained on previous step. We keep adding them to the model.
It is not hard to see that the resulting (possibly infinite) model is indeed the universal model.
We omit the formal proof of this and refer the reader to~\cite{GKKPSZ14AI}.
\end{proof}

So, instead of considering a query $q$ over $D\cup T$ we can consider it over $M_D$.
This observation helps to study rewritings.

It is instructive to consider the graph underlying the model $M_D$. The vertices of the graph are elements of the model and there is a directed edge from an element $m_1$ to an element $m_2$ if there is a binary predicate $P$ such that $M_D \models P(m_1, m_2)$. Then in the process above we start with a graph on constants from $\Delta_D$ and then add new vertices whose existence follows from $T$. Note that the premise of the formula~\eqref{eq:tgd} consists of a unary predicate. This means that the existence of a new element in the model is implied solely by one unary predicate that holds on one of the already constructed vertices. Thus for each new vertex of the model we can trace it down to one of the constants $a$ of the theory and one of the atomic formulas $B(a) \in D$.

The maximal (over all $D$) number of steps of introducing new elements to the model
is called the \emph{depth} of the theory $T$. This parameter will be of interest to us.
We note that $M_D$ and thus the depth of $T$ are not necessarily finite.

In what follows it is useful to consider, for each unary predicate $A \in \Sigma$, the universal model $M_D$ for the theory $D = \{A(a)\}$.
As we mentioned, the universal model for an arbitrary $D$ is ``build up'' from these simple universal models.
We denote this model by $M_A$ (instead of $M_{\{A(a)\}}$) and call it the \emph{universal tree} generated by $A$.
The vertex $a$ in the corresponding graph is called the \emph{root} of the universal tree.
All other vertices of the tree are called \emph{inner vertices}.
To justify the name ``tree'' we note that the underlying graph of $M_A$ in all interesting cases is a tree, though not in all cases.
More precisely, it might be not a tree if some formula~\eqref{eq:tgd} in $T$ does not contain any binary predicate $R(x,y)$.

\begin{example} \label{ex:main}

To illustrate, consider an ontology $T$ describing a part of a student projects organization:
\begin{align*}
& \forall x  \, \big(\nm{Student}(x) \to  \exists y\, (\nm{worksOn}(x,y) \land  \nm{Project}(y))\big),\\
& \forall x  \, \big(\nm{Project}(x) \to \exists y \, (\nm{isManagedBy}(x,y) \land \nm{Professor}(y))\big),\\
& \forall x,y \, \big(\nm{worksOn}(x,y) \to \nm{involves}(y,x)\big), \\
& \forall x,y \, \big(\nm{isManagedBy}(x,y) \to \nm{involves}(x,y)\big).
\end{align*}
Some formulas in this theory are of the form different than~\eqref{eq:tgd}, but it will not be important to us.
Moreover, it is not hard to see that this theory can be reduced to the form~\eqref{eq:tgd}
(along with small changes in data).

Consider the query $q(x)$ asking to find those who work with professors:
\begin{equation}\label{q1}
q(x) = \exists y,z\, \big(\nm{worksOn}(x,y) \land \nm{involves}(y,z) \land  \nm{Professor}(z)\big).
\end{equation}
It is not hard to check that the following formula is a rewriting of $q$:
\begin{multline*}
q'(x) \ \ = \ \
\exists y,z\, \bigl[ \nm{worksOn}(x,y) \land \\
\shoveright{ \bigl(\nm{worksOn}(z,y) \lor\nm{isManagedBy}(y,z) \lor\nm{involves}(y,z)\big) \land \nm{Professor}(z) \bigr] \lor {}}\\
\shoveleft{\exists y\, \bigl[ \nm{worksOn}(x,y) \land \nm{Project}(y) \bigr]  \ \ \lor \ \ \nm{Student}(x).}\\
\end{multline*}
That is, for any data $D$ and any constant $a$ in $D$, we have
$$
D \cup T \models q(a)\  \Leftrightarrow \ I_D \models q'(a).
$$

To illustrate the universal model, consider the data 
$$D = \bigl\{\,\nm{Student}(\nm{c}), \ \nm{worksOn}(\nm{c},\nm{b}), \ \nm{Project}(\nm{b}), \ \nm{isManagedBy}(\nm{b},\nm{a})\,\bigr\}.
$$ 
The universal model $M_D$ is presented in Fig.~\ref{fig:canonical}. The left region corresponds to the data $D$,
the upper right region corresponds to the universal tree generated by $\nm{Project}(\nm{b})$ and the lower right region corresponds to the universal tree generated by $\nm{Student}(\nm{c})$. The label of the form $P^{-}$ on an edge, where $P$ is a predicate of the signature, means that there is an edge in the opposite direction labeled by $P$.

\begin{figure}[ht]
\begin{center}%
\begin{tikzpicture}[>=latex, xscale=1.8, point/.style={circle,draw=black,minimum size=1.5mm,inner sep=0pt}]\small
\fill[rounded corners=16pt,fill=gray!20] (-2.5,-2.2) rectangle +(2.5,3);
\filldraw[fill opacity=0.5,fill=gray!5,draw=black,rounded corners=12pt] (4.1,-2.3) -- (-0.9,-2) -- (-0.9,-1) -- (4.1,-0.8) -- cycle;
\filldraw[fill opacity=0.5,fill=gray!10,draw=black,rounded corners=12pt] (2,-0.6) -- (-0.9,-0.4) -- (-0.9,0.6) -- %(1,0.6) -- (2,1.6) --
(2,0.9) -- cycle;
\draw[dashed,rounded corners=16pt,draw=black] (-2.5,-2.2) rectangle +(2.5,3);
\node at (-2,-1.5) {\large $D$};
\node at (2.5,0.2) {\large $M_{Project}$};
\node at (3.75,-0.52) {\large $M_{Student}$};
\node (b) at (-0.5,-1.5) [point,fill=black,label=above right:{$\nm{c}$},label=below:{$\nm{Student}$}] {};
\node (w1b) at (1.5,-1.5) [point,fill=white,label=above:{\footnotesize$\nm{Project}$}] {};
\draw[->] (b) to node [label=above:{\scriptsize $\nm{worksOn}$},label=below:{\scriptsize $\nm{involves}^-$}] {} (w1b);
\node (w2b) at (3.5,-1.5) [point,fill=white,label=above:{\footnotesize$\nm{Professor}$}] {};
\draw[->] (w1b) to node [label=above:{\scriptsize $\nm{isManagedBy}$},label=below:{\scriptsize $\nm{involves}$}] {} (w2b);
\node (a) at (-0.5,0) [point,fill=black,label=above:{$\nm{Project}$},label=below right:{$\nm{b}$}] {};
\node (c) at (-2,0) [point,fill=black,label=left:{$a$}] {};
\draw[<-] (c) to node [label=above:{\scriptsize $\nm{isManagedBy}\ \ \ \ \ $},label=below:{\scriptsize $\nm{involves}$}] {} (a);
\draw[->] (b) to node [label=left:{\scriptsize \begin{tabular}{c}$\nm{worksOn}$\\$\nm{involves}^-$\end{tabular}}] {} (a);
\node (w1) at (1.5,0) [point,fill=white,label=above:{\footnotesize $\nm{Professor}$}] {};
\draw[->] (a) to node [label=above:{\scriptsize $\nm{isManagedBy}$},label=below:{\scriptsize $\nm{involves}$}] {} (w1);
\end{tikzpicture}%
\end{center}
\caption{An example of a universal model}\label{fig:canonical}
\end{figure}
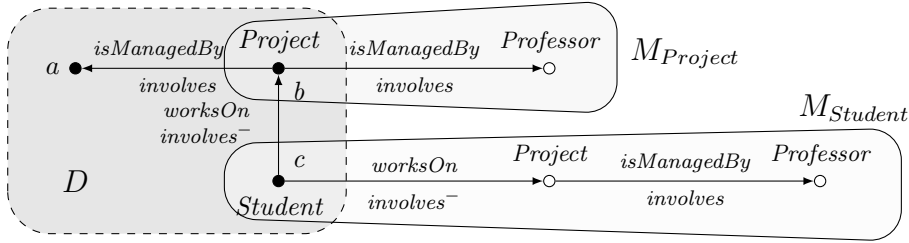

We note that for our query $q(x)$ we have that $q(c)$ follows from $D \cup T$ and we can see that the rewriting $q^\prime(c)$ is true in $M_D$.
Note, however, that $q(c)$ is not true in $D$ due to the incompleteness of the data $D$: it is not known that $a$ is a professor.

\end{example}

From the existence of the universal model (and simplicity of its structure) it can be deduced that for any $q$ there is a rewriting $q^\prime$
having the form of (existentially quantified) disjunction of conjunctions of atoms.
However, it is not hard to provide an example that this rewriting is exponentially larger than $q$ (see~\cite{GKKPSZ14AI}).
By the \emph{size} of the rewriting we mean the number of symbols in the formula.

So to obtain shorter rewriting it is helpful to consider more general types of formulas.
A natural choice would be to allow arbitrary first-order formula as a rewriting. This type is called a first-order rewriting, or a FO-rewriting.
Another option is a positive existential rewriting, or a PE-rewriting. This is a special case of FO-rewriting in which there are no negations and there are only existential quantifiers.
PE-rewritings are more preferable than FO-rewritings since they are more accessible to algorithmic machinery developed for usual databases.
The size of a PE- or a FO-rewriting is a number of symbols in the formula.

Another standard type of rewriting is a nonrecursive datalog rewriting, or NDL-rewriting.
This rewriting does not have a form of first-order formula and instead has the form of DAG-representation of a first-order formula.
Namely, NDL-rewriting consists of the set $\Pi$ of formulas of the form
$$
\forall \vec{x} \left(A_1 \wedge \ldots \wedge A_n \rightarrow A_0\right),
$$
where $A_i$ are atomic formulas (possibly new, not presented in the original signature $\Sigma$) not necessarily of arity $1$ or $2$.
Each $A_i$ depends on (some of) the variables from $\vec{x}$ and each variable in $A_0$ must occur in $A_1 \wedge \ldots \wedge A_n$.
Finally, we need the acyclicity property of $\Pi$. To define it, consider a directed graph whose vertices are predicates $A$ of $\Pi$
and there is an edge from $A$ to $B$ iff there is a formula in $\Pi$ which has $B$ as the right-hand side and contains $A$ in the left-hand side.
Now П is called \emph{acyclic} if the resulting graph is acyclic.
Also an NDL-rewriting contains a goal predicate $G$ and we say that $\vec{a}$ in $\Delta_D$ satisfies $(\Pi, G)$ over the data $D$ iff
$$
D \cup \Pi \models G(\vec{a}).
$$
Thus, a $(\Pi,G)$ is called an NDL-rewriting of the query $q$ if
$$
D \cup T \models q(\vec{a})\  \Leftrightarrow \ D \cup \Pi \models G(\vec{a})
$$
for all $D$ and all $\vec{a}$.
The \emph{size} of an NDL-rewriting $(\Pi, G)$ is the number of symbols in it.

\begin{example} To illustrate the concept of NDL-rewriting we provide explicitly a rewriting for the query $q$ from Example~\ref{ex:main}:
\begin{align*}
& \forall y,z  \left(\nm{worksOn}(z,y) \rightarrow N_1(y,z)\right),\\
& \forall y,z  \left(\nm{isManagedBy}(y,z) \rightarrow N_1(y,z)\right),\\
& \forall y,z  \left(\nm{involves}(y,z) \rightarrow N_1(y,z)\right),\\
& \forall x, y, z   \left(\nm{worksOn}(x,y) \land N_1(y,z) \land \nm{Professor}(z) \rightarrow G(x)\right), \\
& \forall x, y   \left(\nm{worksOn}(x,y) \land \nm{Project}(y) \rightarrow G(x)\right), \\
& \forall x   \left(\nm{Student}(x) \rightarrow G(x)\right),
\end{align*}
where $N_1$ is a new binary predicate and $G$ is the goal predicate of this NDL-rewriting.

It is not hard to see that this rewriting is similar to the PE-rewriting $q^\prime$ from Example~\ref{ex:main}.
Indeed, $N_1(y,z)$ is equivalent to the subformula
$$
\bigl(\nm{worksOn}(z,y) \lor\nm{isManagedBy}(y,z) \lor\nm{involves}(y,z)\big)
$$
of $q^\prime$ and $G(x)$ is equivalent to $q^\prime(x)$.
\end{example}

It turns out that NDL-rewritings are more general than PE-rewritings.
Indeed, a PE-rewriting $q^\prime$ has the form $\exists \vec{y} \phi(\vec{x},\vec{y})$,
where $\phi$ is a monotone Boolean formula applied to atomic formulas (note that the existential quantifiers can be moved to the prefix due to the fact that there are no negations in the formula).
The formulas in $\Pi$ can model $\lor$ and $\land$ operations and thus can model the whole formula $\phi$. For this,
for each subformula of $\phi$ we introduce a new predicate symbol that depends on all variables on which this subformula depends.
We model $\lor$ and $\land$ operations on subformulas one by one. In the end we will have an atom $F(\vec{x},\vec{y})$. Finally, we add to $\Pi$ the formula
\begin{equation} \label{eq:final}
\forall \vec{x},\vec{y}\ \big(F(\vec{x},\vec{y}) \rightarrow G(\vec{x})\big).
\end{equation}
Then we have that, for any $\vec{a},\vec{b}$, $\phi(\vec{a},\vec{b})$ is true on $D$ iff $F(\vec{a},\vec{b})$ is true over $D \cup \Pi$.
Finally, $\exists \vec{y} \phi(\vec{a},\vec{y})$ is true over $D$ iff there is $\vec{b}$ among constants such that $\phi(\vec{a},\vec{b})$ is true.
On the other hand, in $\Pi$ we can deduce $G(\vec{a})$ iff there is $\vec{b}$ such that $F(\vec{a},\vec{b})$ is true.
Thus, given a PE-rewriting, we can construct an NDL-rewriting of approximately the same size.

It is unknown whether NDL-rewritings and FO-rewritings are comparable. On the one hand, NDL-rewritings correspond to Boolean circuits
and FO-rewritings --- to Boolean formulas. On the other hand, FO-rewritings can use negations and NDL-rewritings are monotone.

As we said above, we will consider only conjunctive queries $q(x)$ to knowledge bases.
However, in many cases queries have even simpler structure.
To describe these restricted classes of queries, we have to consider a graph underlying the query.
The vertices of the graph are variables appearing in $q$.
Two vertices are connected iff their labels appear in the same atom of $q$.
If this graph is a tree we call a query \emph{tree-like}.
If the graph is a path, then we call a query \emph{linear}.

\section{Rewriting size lower bounds: general approach} \label{sec:proof_idea}

In this section we will describe the main idea behind the proofs of lower bounds on the size of query rewritings.

Very informally, we encode Boolean functions inside of queries in such a way that the rewritings
correspond to Boolean circuits computing these functions. If we manage to encode hard enough function,
then there will be no small circuits for them and thus there will be no small rewritings.

How exactly do we encode functions inside of queries?
First of all we will restrict ourselves to the data $D$ with only one constant element $a$.
This is a substantial restriction on the data. But since our rewritings should work for any data
and we are proving lower bounds, we can make our task only harder. On the other hand, this restriction makes our lower bounds
more general.

Next, we introduce several unary predicates $A_1, A_2, \ldots, A_n$
and consider the formulas $A_i(a)$. These predicates correspond to Boolean variables $x_1, \ldots, x_n$
of encoded function $f$: the variable $x_i$ is true iff $A_i(a) \in D$.
There are other predicates in the signature and other formulas in $D$.
Their role would be to make sure that
$$
D \cup T \models q(x)
$$
iff the encoded function $f$ is true on the corresponding input.

This approach allows us to characterize the expressive power of various queries and theories.
This characterization is summarized in the following table.

\begin{center}
\begin{tabular}{c|c|c|c}
  % after \\: \hline or \cline{col1-col2} \cline{col3-col4} ...
  & depth $1$ & depth $d>1$ & arbitrary depth \\ \hline
  linear queries & $\leq\NC^1$~\cite{LICS14} & $\NL/\poly$~\cite{DL14} & $\NL/\poly$~\cite{DL14} \\ \hline
  tree-like queries & $\leq\NC^1$~\cite{LICS14} & $\SAC^1$~\cite{DL14} & $\NP/\poly$~\cite{GKKPSZ14AI} \\ \hline
  general queries & $\NL/\poly$~\cite{LICS14}  & $\NP/\poly$~\cite{LICS14} & $\NP/\poly$~\cite{GKKPSZ14AI} \\
\end{tabular}
\end{center}

The columns of the table correspond to the classes of the theories $T$.
The rows of the table correspond to the classes of the queries $q$.
An entry of the table represents the class of functions
that can be encoded by queries and theories of these types.
The results in the table give both upper and lower bounds.
However, in what follows we will concentrate on lower bounds,
that is we will be interested in how to encode hard functions
and we will not discuss why harder functions cannot be encoded.

Next, we need to consider a rewriting of one of the types described above and obtain from it the corresponding computational model computing $f$.
This connection is rather intuitive: rewritings has a structure very similar to certain types of Boolean circuits.
Namely, FO-rewritings are similar to Boolean formulas, PE-rewriting are similar to monotone Boolean formulas
and NDL-rewritings are similar to monotone Boolean circuits.
Thus, polynomial size FO-rewriting means that $f$ is in $\NC^1$, polynomial size PE-rewriting means that $f$ is in $\mNC^1$, and
polynomial size NDL-rewriting means that $f$ is in $\mP/\poly$.
We omit the proofs of these reductions.

Together with the table above this gives the whole spectrum of results on the size of rewritings.
We just need to use the results on the relations between corresponding complexity classes.
For example, in case of depth $1$ theories and path-like or tree-like queries there are polynomial rewritings of all three types.
In case of depth 2 theory and path-like or tree-like queries there are no polynomial PE-rewriting, there are no polynomial FO-rewritings
under certain complexity-theoretic assumption, but there are polynomial NDL-rewritings.
In case of depth 2 theories and arbitrary queries there are no polynomial PE- and NDL-rewritings and there are no polynomial FO-rewritings
under certain complexity-theoretic assumption.

Below we provide further details of the proofs of aforementioned results.
The paper~\cite{GKKPSZ14AI} used an add-hoc construction to deal with the case of unbounded depth and non-linear queries.
Subsequent papers~\cite{LICS14,DL14} provided a unified approach that uses the so-called hypergraph programs.

In the next section we proceed to the discussion of these programs.

\section{Hypergraph programs: origination} \label{sec:hypergraph_intro}

For the sake of simplicity we will restrict ourselves to Boolean queries only.
Consider a query $q = \exists \vec{y} \phi(\vec{y})$ and consider its underlying graph $G$.
Vertices of $G$ correspond to the variables of $q$.
Directed edges of $G$ correspond to binary predicates in $q$.
Each edge $(u,v)$ is labeled by all atomic formulas $P(u,v)$ in $q$.
Each vertex $v$ is labeled by $A$ if $A(v)$ is in $q$.
%We also distinguish vertices corresponding to $x$ (real vertices) and vertices corresponding to $y$ (existential).

Let us consider data $D$. We can construct a universal model $M_D$ just by adding universal trees to each element of $D$.
Let us see how the query can be satisfied by elements of the universal model.
For this we need that for each variable $t$ of the query we find a corresponding element in $M_D$ satisfying all the properties of $t$ stated in the query.
This element in $M_D$ can be an element of the data and also can be an element of universal trees.
%However, real vertices should correspond to elements of the data.

Thus, for a query $q$ to be satisfied we need an embedding of it into the universal model.
That is we should map vertices of $G$ into the vertices of the universal model $M_D$ in such a way that for each label
in $G$ there is a corresponding label in $M_D$. We call this embedding a \emph{homomorphism}.

Now let us see how a vertex $v$ of $G$ can be mapped into an inner element $w$ of a universal tree $R$.
This means that for all labels of $v$ the vertex $w$ in a universal tree $R$ should have the same labels and
for all adjacent edges of $v$ there should be corresponding edges adjacent to $w$ in a universal tree.
Thus all vertices adjacent to $v$ should be also mapped in the universal tree $R$.
We can repeat this argument for the neighbors of $v$ and proceed until we reach vertices of $G$ mapped into the root of $R$.
So, if one of the vertices of $G$ is embedded into a universal tree $R$, then so is a set of neighboring vertices.
The boundary of this set of vertices should be mapped into the root of the universal tree.

Let us summarize what we have now. An answer to a query corresponds to an embedding of $G$ into the universal model $M_D$.
There are connected induced subgraphs in $G$ that are embedded into universal trees.
The boundaries of these subgraphs (the vertices connected to the outside vertices) are mapped into the root of the universal tree.
Two subgraphs can intersect only by boundary vertices. These subgraphs are called \emph{tree witnesses}.

Given a query we can find all possible tree witnesses in it. Then, for any given data $D$ there is an answer to the query
if we can map the query into the universal model $M_D$. There is such a mapping if we can find a set of disjoint tree witnesses
such that we can map all other vertices into $D$ and the tree witnesses into the corresponding universal trees.

Now assume for simplicity that there is only one element $a$ in $D$. Thus $D$ consists of formulas $A(a)$ and $P(a,a)$.
To decide whether there is an answer to a query we need to check whether there is a set of tree witnesses which do not intersect (except by boundary vertices),
such that all vertices except the inner vertices of tree witnesses can be mapped in $a$.
Consider the following hypergraph $H$: it has a vertex for each vertex of $G$ and for each edge of $G$; for each tree witness there is a hyperedge in $H$ consisting of vertices corresponding to the inner vertices of the tree witness and of vertices corresponding to the edges of the tree witness.
For each vertex $v$ of the hypergraph $H$ let us introduce a Boolean variable $x_v$ and for each hyperedge $e$ of the hypergraph $H$ --- a Boolean variable $x_e$.
For a given $D$ (with one element $a$) let $x_v$ be equal to $1$ iff $v$ can be mapped in $a$ and  let $x_e$ be equal to $1$ iff the unary predicate generating the tree witness corresponding to the hyperedge $e$ is true on $a$.
From the discussion above it follows that there is an answer to a rewriting for a given $D$ iff there is a subset of disjoint hyperedges such that $x_e=1$ for them and they contain all vertices with $x_v=0$.

This leads us to the following definition.
\begin{definition}[Hypergraph program]
A hypergraph program $H$ is a hypergraph whose vertices are labeled by Boolean variables $x_1, \ldots, x_n$, their negations or Boolean constants $0$ and $1$.
A hypergraph program $H$ outputs $1$ on input $\vec{x}\in \{0,1\}^n$ iff there is a set of disjoint hyperedges covering all vertices whose labels evaluates to $0$. We denote this by $H(\vec{x})=1$. A hypergraph program computes a Boolean function $f \colon \{0,1\}^n \to \{0,1\}$ iff for all $\vec{x} \in \{0,1\}^n$ we have $H(\vec{x})=f(\vec{x})$. The size of a hypergraph program is the number of vertices plus the number of hyperedges in it.
A hypergraph program is monotone iff there are no negated variables among its labels.
\end{definition}

\begin{remark}
Note that in the discussion above we obtained somewhat different model. Namely, there were also variables associated to hyperedges of the hypergraph. Note, however, that our definition captures also this extended model. Indeed, we can introduce for each hyperedge $e$ a couple of new fresh vertices $v_e$ and $u_e$ and a new hyperedge $e^\prime$. We add $v_e$ to the hyperedge $e$ and we let $e^\prime = \{v_e, u_e\}$. The label of $v_e$ is $1$ and the label of $u_e$ is the variable $x_e$. It is easy to see that $x_e=0$ iff we cannot use the hyperedge $e$ in our cover.
\end{remark}

So far we have discussed how to encode a Boolean function by a query and a theory.
We have noted that the resulting function is computable by a hypergraph program.
We denote by $\HGP$ the class of functions computable by hypergraph programs of polynomial size
(recall, that we actually consider sequences of functions and sequences of programs).

Various restrictions on queries and theories result in restricted versions of hypergraph programs.
If a theory is of depth $1$, then each tree witness has one inner vertex and thus
two different hyperedges can intersect only by one vertex corresponding to the edge of $G$.
Thus each vertex corresponding to the edge of $G$ can occur in at most two hyperedges and
the resulting hypergraph program is of degree at most $2$.
We denote by $\HGP_k$ the set of functions computable by polynomial size hypergraph programs of degree at most $k$.

If a query is tree-like (or linear), then the hypergraph program will have an underlying tree (or path) structure
and all hyperedges will be its subtrees (subpaths).
We denote by $\HGPt$ ($\HGPp$) the set of functions computable by hypergraph programs of polynomial size and with underlying tree (path) structure.

However, to prove lower bounds we need to show that \emph{any} hypergraph program in certain class can be encoded
by a query and a theory of the corresponding type.
These statements are proved separately by various constructions of queries and theories.
We will describe a construction for general hypergraph programs as an example.

Consider a hypergraph program $P$  and consider its underlying hypergraph $H = (V,E)$.
It would be more convenient to consider a more general hypergraph program $P^{\prime}$
which has the same underlying hypergraph $H$ and each vertex $v in V$ is labeled by a variable $x_v$.
Clearly, the function computed by $P$ can be obtained from the function computed by $P^\prime$
by fixing some variables to constant and identifying some variables (possibly with negations).
Thus it is enough to encode in a query and a theory the function computed by $P^\prime$.
We denote this function by $f$.

To construct a theory and a query encoding $f$
consider the following directed graph $G$. It has a vertex $z_v$ for each vertex $v$ of the hypergraph $H$
and a vertex $z_e$ for each hyperedge $e$ of the hypergraph $H$.
The set of edges of $G$ consists of edges $(z_v, z_e)$ for all pairs $(v,e)$ such that $v \in e$.
This graph will be the underlying graph of the query.
For each vertex $z_e$ the subgraph induced by all vertices on the distance at most $2$ from $z_e$ will be a tree witness.
In other words, this tree witness contains vertices $z_v$ for all $v \in e$ and $z_{e^\prime}$ for all $e^\prime$ such that $e^\prime \cap e \neq \emptyset$.
The latter vertices are boundary vertices of the tree witness.

The signature contains unary predicates $A_v$ for all $v \in V$, unary predicates $A_e, B_e$ and binary predicates $R_e$ for all $e \in E$.
Intuitively, the predicate $A_e$ generates tree-witness corresponding to $z_e$,
the predicate $B_e$ encodes that its input correspond to $z_v$ with $v \in e$,
the predicate $R_e$ encodes that its inputs correspond to $(z_e, z_v)$ and $v \in e$,
the predicate $A_v$ encodes the variable $x_v$ of $f$.

Our Boolean query $q$ consists of atomic formulas
$$
\{ A_v(z_v) \mid v \in V \} \cup \{ R_e(z_e, z_v) \mid v \in e, \text{ for } v \in V \text{ and } e \in E \}.
$$
Here $z_v$ and $z_e$ for all $v \in V$ and $e \in E$ are existentially quantified variables of the query.

Theory $T$ consists of the following formulas (the variable $x$ is universally quantified):
$$
A_e(x) \to \exists y \bigwedge_{\begin{subarray}{c}e \cap e' \neq \emptyset\\ e \ne e'\end{subarray}} \big( R_{e'}(x,y) \land
B_e(y) \big),\ \
B_e(x) \to \bigwedge_{v \in e} A_v(x), \ \  B_e(x) \to \exists y R_e(y,x).
$$
In particular, each predicate $A_e$ generates a universal tree of depth 2 consisting of 3 vertices $a, w_{vertex}^e, w_{edge}^e$
and of the following predicates ($a$ is a root of the universal tree):

\begin{minipage}{0.45\textwidth}
\begin{align*}
&A_{e}(a), \\
&R_{e^\prime}(a,w_{vertex}^e) \text{ for all } e^\prime \neq e,\ e^\prime \cap e \neq \emptyset,\\
&B_{e}(w_{vertex}^e),\\
&A_{v}(w_{vertex}^e)\text{ for all } v \in e,\\
&R_e(w_{edge}^e,w_{vertex}^e).
\end{align*}
\end{minipage}
\hfill
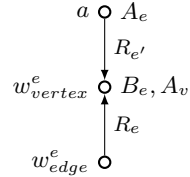
\begin{minipage}{0.45\textwidth}
\begin{center}
\begin{tikzpicture}[>=latex, point/.style={circle,draw=black,thick,minimum size=1.5mm,inner sep=0pt}, wiggly/.style={thick,decorate,decoration={snake,amplitude=0.3mm,segment length=2mm,post length=1mm}},
query/.style={thick},tw/.style={shorten <= 0.1cm, shorten >= 0.1cm,dashed},yscale=1,xscale=0.8]\footnotesize
\node (a) at (0,2) [point, label=right:{$A_{e}$}, label=left:{$a$}]{};
\node (d1) at (0,1) [point, fill=white,label=left:$w_{vertex}^e$,label=right:{$B_{e},A_{v}$}]{};
\node (d2) at (0,0) [point, fill=white,label=left:$w_{edge}^e$]{};
\draw[->] (a)  to node [right]{\scriptsize $R_{e^\prime}$} (d1);
\draw[->] (d2)  to node [right]{\scriptsize $R_{e}$} (d1);
\end{tikzpicture}
\end{center}
\end{minipage}
\medskip

There are other universal trees generated by predicates $B_e$, but we will consider only data in which $B_e$ are not presented, so the corresponding universal  trees also will not be presented in the universal model.

There is one constant $a$ in our data and we will restrict ourselves only to the data containing $A_e(a)$ for all $e$
and $R_{e}(a,a)$ for all $e$ and not containing $B_{e}$ for all $e$. For convenience denote $D_0 = \{A_e(a), R_e(a,a) \text{ for all } e \in E\}$.
The predicates $A_v$ will correspond to the variables $x_v$ of the function $f$.
That is the following claim holds.
\begin{claim}
For all $\vec{x} \in \{0,1\}^{n}$ $f(\vec{x}) = 1$ iff $D\cup T \models q$ for $D = D_0 \cup \{A_v(a) \mid x_v  = 1\}$.
\end{claim}
\begin{proof}
Note first that if $A_v(a)$ is true for all $v$ then the query is satisfiable.
We can just map all vertices $z_e$ and $z_v$ to $a$. However, if some predicate $A_v(a)$ is not presented,
then we cannot map $z_v$ to $a$ and have to use universal trees.

Suppose $f(\vec{x})=1$ for some $\vec{x} \in \{0,1\}^n$ and consider the corresponding data $D$.
There is a subset of hyperedges $E^\prime \subseteq E$ of $H$ such that hyperedges in $E^\prime$ do not intersect and
all $v \in V$ such that $x_v = 0$ lie in hyperedges of $E^\prime$.
Then we can satisfy the query in the following way. We map the vertices $z_e$ with $e \notin E^\prime$ to $a$.
We map all vertices $z_v$ such that $v$ is not contained in hyperedges of $E^\prime$ also into $a$.
If for $z_v$ we have $v \in e$ for $e \in E^\prime$, then we send $z_v$ to the $w_{vertex}^e$ vertex in the universal tree $M_{A_e}$.
Finally, we send vertices $z_e$ with $e \in E^\prime$ to $w_{edge}^e$ vertex of the universal tree $M_{A_e}$.
It is easy to see that all predicates in the query are satisfied.

In the other direction, suppose for data $D$ the query $q$ is true.
It means that there is a mapping of variables $z_v$ and $z_e$ for all $v$ and $e$ into universal model $M_D$.
Note that the vertex $z_e$ can be sent either to $a$, or to the vertex $w_{edge}^e$ in the universal tree $M_{A_{e}}$.
Indeed, only these vertices of $M_D$ has outgoing edge labeled by $R_e$.
Consider the set $E^\prime = \{ e \in E \mid z_e \text{ is sent to } w_{edge}^e\}$.
Consider some $e \in E^\prime$ and note that for any $e^\prime$, such that $e^\prime \neq e$ and $e^\prime \cap e \neq \emptyset$, $z_{e^\prime}$ is on
the distance $2$ from $z_e$ in $G$ and $z_{e^\prime}$ should be mapped in $a$.
Thus hyperedges in $E^\prime$ are non-intersecting.
If for some $z_v$ the atom $A_v(a)$ is not in $D$, then $z_v$ cannot be mapped into $a$.
Thus it is mapped in the vertex $w_{vertex}^e$ in some $M_{A_e}$ for some $e$ containing $v$.
But then $z_e$ should be mapped into $w_{edge}^e$ of the same universal tree (there is only one edge leaving $w_{vertex}^e$
labeled by $R_e$). Thus $e \in E^\prime$ and thus $v$ is covered by hyperedges of $E^\prime$.
Overall, we have that hyperedges in $E^\prime$ give a disjoint cover of all zeros in $P^\prime$ and thus $f(x)=1$.

\end{proof}

\section{Hypergraph programs: complexity} \label{sec:hypergraph_complexity}

We have discussed that hypergraph programs can be encoded by queries and theories.
Now we need to show that there are hard functions computable by hypergraph programs.
For this we will determine the power of various types of hypergraph programs.
Then the existence of hard functions will follow from known results in complexity theory.

We formulate the results on the complexity of hypergraph programs in the following theorem.

\begin{theorem}[\cite{LICS14,DL14}]\label{thm:hgp_complexity}
The following equations hold both in monotone and non-monotone cases:
\begin{enumerate}
  \item $\HGP = \HGP_3 = \NP/\poly$;
  \item $\HGP_2 = \NBP$;
  \item $\HGPp = \NBP$;
  \item $\HGPt = \SAC^1$.
\end{enumerate}
\end{theorem}

Together with the discussion of two previous sections this theorem gives the whole picture of proofs of lower bounds on the rewriting size
for considered types of queries and theories.

We do not give a complete proof of Theorem~\ref{thm:hgp_complexity} here,
but in order to present ideas behind it, we give a proof of the first part of the theorem.

\begin{proof}

Clearly, $\HGP_3 \subseteq \HGP$.

Next, we show that $\HGP \subseteq \NP/\poly$.
Suppose we have a hypergraph program of size $m$ with variables $\vec{x}$.
We construct a circuit $C(\vec{x},\vec{y})$ of size $\poly(m)$ satisfying~\eqref{eq:NP_poly}.
Its $\vec{x}$-variables are precisely the variables of the program, and certificate variables $\vec{y}$ correspond to the hyperedges of the program.
The circuit $C$ will output 1 on $(\vec{x},\vec{y})$ iff the family $\{e \mid y_{e} = 1\}$ of hyperedges of the hypergraph
forms a disjoint set of hyperedges covering all vertices labeled by $0$ under $\vec{x}$.
It is easy to construct a polynomial size circuit checking this property.
Indeed, for each pair of intersecting hyperedges $(e, e^\prime)$ it is enough  to compute disjunction $\neg y_e \vee \neg y_{e^\prime}$,
and for each vertex $v$ of the hypergraph with label $t$ and contained in hyperedges $e_{1}, \ldots, e_{k}$
it is enough to compute disjunction $t \vee  y_{e_{1}} \lor \dots \lor y_{e_{k}}$.
It then remains to compute a conjunction of these disjunctions.
It is easy to see that this construction works also in monotone case
(note that applications of $\neg$ to $\vec{y}$-variables in the monotone counterpart of $\NP/\poly$ are allowed).

Now we show that $\NP/\poly \subseteq \HGP_3$.
Consider a function $f \in \NP/\poly$ and consider a circuit $C(\vec{x},\vec{y})$ satisfying~\eqref{eq:NP_poly}.
Let $g_1,\dots,g_n$ be the gates of $C$ (including the inputs $\vec{x}$ and $\vec{y}$). We construct a hypergraph program of degree $\le 3$ computing $f$ of size polynomial in the size of $C$. For each $i$ we introduce a vertex $g_i$ labelled with $0$ and a pair of hyperedges $\bar{e}_{g_i}$ and $e_{g_i}$, both containing $g_i$. No other hyperedge contains $g_i$, and so either $\bar{e}_{g_i}$ or $e_{g_i}$ should be present in any cover of zeros in the hypergraph program. Intuitively, if the gate $g_i$ evaluates to $1$ then $e_{g_i}$ is in the cover, otherwise $\bar{e}_{g_i}$ is there. To ensure this property for each input variable $x_i$, we add a new vertex $v_i$ labelled with $\neg x_i$ to $e_{x_i}$  and a new vertex $u_i$ labelled with $x_i$ to $\bar{e}_{x_i}$.
For a non-variable gate $g_i$, we consider three cases.
\begin{itemize}
\item[--] If $g_i = \neg g_j$ then we add a vertex labelled with $1$ to $e_{g_i}$ and $\bar{e}_{g_j}$, and a vertex labelled with $1$ to $\bar{e}_{g_i}$ and $e_{g_j}$.

\item[--] If $g_i = g_j \vee g_{j'}$ then we add a vertex labelled with $1$ to $e_{g_j}$ and $\bar{e}_{g_i}$,
add a vertex labelled with $1$ to $e_{g_{j'}}$ and $\bar{e}_{g_i}$;
then, we add vertices $h_j$ and $h_{j'}$ labelled with $1$ to $\bar{e}_{g_j}$ and $\bar{e}_{g_{j'}}$, respectively, and a vertex $w_{i}$ labeled with $0$ to $\bar{e}_{g_i}$; finally, we add hyperedges $\{h_j, w_i\}$ and $\{h_{j'}, w_i\}$.

\item[--] If $g_i = g_j \wedge g_{j'}$ then we use the dual  construction.
\end{itemize}
In the first case it is not hard to see that $e_{g_i}$ is in the cover iff $\bar{e}_{g_j}$ is in the cover.
In the second case $e_{g_i}$ is in the cover iff at least one of $e_{g_j}$ and $e_{g_{j'}}$ is in the cover.
Indeed, in the second case if, say, the cover contains $e_{g_j}$  then it cannot contain $\bar{e}_{g_i}$, and so it contains $e_{g_i}$. The vertex $w_i$ in this case can be covered by the hyperedge $\{h_j, w_i\}$ since $\bar{e}_{g_j}$ is not in the cover.
Conversely, if neither $e_{g_j}$ nor $e_{g_{j'}}$ is in the cover, then it must contain both $\bar{e}_{g_j}$ and
$\bar{e}_{g_{j'}}$ and so, neither $\{h_j, w_i\}$ nor $\{h_{j'}, w_i\}$ can belong to the cover
and we will have to include $\bar{e}_{g_i}$ to the cover.
Finally, we add one more vertex labelled with $0$ to $e_{g}$ for the output gate $g$ of $C$. It is not hard to show that, for each $\vec{x}$, there is $\vec{y}$ such that $C(\vec{x},\vec{y})=1$ iff the constructed hypergraph program returns 1 on $\vec{x}$.

For the monotone case, we remove all vertices labelled with $\neg x_i$. Then, for an input $\vec{x}$, there is a cover of zeros  in the resulting hypergraph program iff there are $\vec{y}$ and $\vec{x}' \leq \vec{x}$ with $C(\vec{x}',\vec{y})=1$.
\end{proof}

\section*{Acknowledgments} The author is grateful to Michael Zakharyaschev, Mikhail Vyalyi, Evgeny Zolin and Stanislav Kikot for helpful comments on the preliminary version of this survey.

{\small
\bibliographystyle{abbrv}
\bibliography{OBDA}
}

\end{document}